\documentclass{IEEEtran}  
\usepackage{amssymb}
\usepackage{amsmath}
\usepackage{amsthm}
\usepackage{graphicx}
\usepackage{comment}
\usepackage{booktabs}
\usepackage{balance}
\usepackage{bm}
\usepackage{graphicx}
\usepackage{amssymb}
\usepackage{wasysym}
\usepackage{float}
\usepackage{color}
\usepackage{dsfont}
\usepackage{algorithm,algorithmicx}
\usepackage{algpseudocode}
\usepackage{epstopdf}
\usepackage{tabularx}
\usepackage{cite}
\usepackage{enumerate}
\usepackage{empheq}
\usepackage{subcaption}
\usepackage{eurosym}
\usepackage{marginnote}

\newtheorem{proposition}{Proposition}

\newtheorem{remark}{Remark}
\newtheorem{assumption}{Assumption}

\newcommand{\eod}{\ensuremath{\hfill\Box}}

\makeatletter
\newcommand\fs@betterruled{%
	\def\@fs@cfont{\bfseries}\let\@fs@capt\floatc@ruled
	\def\@fs@pre{\vspace*{5pt}\hrule height.8pt depth0pt \kern2pt}%
	\def\@fs@post{\kern2pt\hrule\relax}%
	\def\@fs@mid{\kern2pt\hrule\kern2pt}%
	\let\@fs@iftopcapt\iftrue}
\floatstyle{betterruled}
\restylefloat{algorithm}
\makeatother

\newcommand{\bs}{\boldsymbol}
\newcommand{\mc}{\mathcal}
\newcommand{\bb}{\mathbb}
\newcommand{\R}{\bb R}
\newcommand{\avg}{\textrm{avg}}

\newcommand{\argmin}{\operatorname{argmin}}

\newcommand{\proj}{\mathrm{proj}}

\newcommand{\diag}{\operatorname{diag}}

\newcommand{\col}{\operatorname{col}}

\newcommand{\0}{\mathbf{0}}
\newcommand{\1}{\mathbf{1}}

\title{
Operationally-Safe Peer-to-Peer Energy Trading in Distribution Grids: A Game-Theoretic Market-Clearing Mechanism
}
\author{Giuseppe Belgioioso, Wicak Ananduta, Sergio Grammatico, \emph{Senior Member, IEEE,} \\ and Carlos Ocampo-Martinez, \emph{Senior Member, IEEE}
\thanks{
	G. Belgioioso is with the Automatic Control Laboratory, ETH Zurich, Switzerland. W. Ananduta and S. Grammatico are with the Delft Center for Systems and Control (DCSC), TU Delft, The Netherlands. C. Ocampo-Martinez is with the Automatic Control Department, Universitat Polit\`{e}cnica de Catalunya, Barcelona, Spain. 
	E-mail addresses: \texttt{gbelgioioso@ethz.ch}, \texttt{carlos.ocampo@upc.edu} \texttt{\{w.ananduta,s.grammatico\}@tudelft.nl}. This work was partially supported by NWO under research projects OMEGA (grant n. 613.001.702), P2P-TALES (grant n. 647.003.003), the ERC under research project COSMOS (802348) and PID2020-115905RB-C21 (L-BEST) funded by MCIN/ AEI /10.13039/501100011033. 
}
}

\begin{document}

\maketitle

\begin{abstract}
In future distribution grids, prosumers (i.e., energy consumers with storage and/or production capabilities) will trade energy with each other and with the main grid.
To ensure an efficient and safe operation of energy trading, in this paper, we formulate a peer-to-peer energy market of prosumers as a generalized aggregative game, in which a network operator is responsible to enforce the operational constraints of the system.
We design a distributed market-clearing mechanism with convergence guarantee to an economically-efficient, strategically-stable, and operationally-safe configuration (i.e., a variational generalized Nash equilibrium). 
Numerical studies on the IEEE 37-bus testcase show the scalability of the proposed approach and suggest that active participation in the market is beneficial for both prosumers and the network operator.  
\end{abstract}
\vspace{-5pt}
\vspace{-10pt}
\section*{Nomenclature}
\begin{table}[h]
	\small 
	\begin{tabular}{l c l}
		\toprule
		\multicolumn{3}{c}{Variables and Cost Functions} \\
		\hline
		$f^{\mathrm{di}}$ & [\euro] & cost of the \underline{di}spatchable units  \\
		$f^{\mathrm{mg}}$ & [\euro] & cost of trading with the \underline{m}ain \underline{g}rid  \\
		$f^{\mathrm{st}}$ & [\euro] & cost  of the \underline{st}orage units  \\
		$f^{\mathrm{tr}}$ & [\euro]& cost  of \underline{tr}ading with other prosumers  \\
		$J$ &[\euro]  & total cost function of each prosumer \\
		$\lambda^{\mathrm{mg}}$ &[\euro/kWh] & dual variable for grid trading constraints 		\\
		$\mu^{\mathrm{pb}}$ & [\euro/kWh]& dual variable for power balance constraints \\
		$\mu^{\mathrm{tg}}$ &[\euro/kWh] & dual variable for grid physical constraints
		\\
		$\mu^{\mathrm{tr}}$ &[\euro/kWh] & dual variable for reciprocity constraints \\
		$p^{\mathrm{di}}$ & [kW]& power generated by {di}spatchable units  \\
		$p^{\ell}$ &[kW] & real power line of two neighboring busses  \\
		$p^{\mathrm{mg}}$ &[kW] & power traded with the {m}ain {g}rid  \\
		{$p^{\mathrm{ch}}$} &[kW] & {charging power of the {st}orage units}  \\
		{$p^{\mathrm{ds}}$} &[kW] & {discharging power of the {st}orage units}  \\
		$p^{\mathrm{tg}}$ &[kW] & power exchanged between bus and main grid \\	
		$p^{\mathrm{tr}}$ & [kW]& power {tr}aded with another prosumer  \\
		$q^{\ell}$ &[kVAr] & reactive power line   \\
		$\sigma^{\mathrm{mg}}$ &[kW] & aggregate of active load on the main grid \\
		$v$ &p.u. & voltage magnitude \\
		$x$ & [$\%$] & state of charge of the storage units\\
		$\theta$ & [rad]& voltage angle \\
		\hline 

	\end{tabular}
\end{table} 
\smallskip
\begin{table}[!h]
	\small 
	\begin{tabular}{l c l}
				\toprule
		\multicolumn{3}{c}{Parameters} \\
		\hline
		$\alpha, \beta, \gamma$ &- & step sizes of the proposed algorithm \\
		$b$ & [kW]& aggregate of passive consumer demand \\
		$B$ &[ohm$^{-1}$] & line susceptance \\
		$ c^{\mathrm{di}}$ &[\euro/kWh] & linear coefficient (coeff.) on the cost \\ & & of dispatchable units (DU) \\
		$c^{\mathrm{ta}}$ & [\euro/kWh]& trading tariff \\
		$c^{\mathrm{tr}}$ &[\euro/kWh] & per-unit cost of trading \\
		$d^{\mathrm{mg}}$ & [\euro/kWh$^2$]& coeff. on the cost of trading with \\
		& & the main grid \\
		$e^{\mathrm{cap}}$ &[kWh] & max. capacity of the storage units \\
		{$\eta^{\mathrm{st}}$} & - & {leakage coefficient of storage units}\\
		{$ \eta^{\mathrm{ch}}$}, {$ \eta^{\mathrm{ds}}$} & - & {charging and discharging efficiencies}\\
		$G$ & [ohm$^{-1}$]  & line conductance \\
		$H$ &- & time horizon \\
		$\bar p^{\mathrm{ch}}$ &[kW]  & max. \underline{ch}arging power of the storages \\
		$p^{\mathrm{d}}$ &[kW]  & power \underline{d}emand 
		\\
		{$\bar p^{\mathrm{ds}}$} &[kW] & max. \underline{d}i\underline{s}c{h}arging power of the storages \\
		$\overline{p}^{\mathrm{di}}, \underline{p}^{\mathrm{di}}$ &[kW]  & max. and min. power generated by DU\\
		$\overline{p}^{\mathrm{mg}}, \underline{p}^{\mathrm{mg}}$ &[kW]  & max. and min. total power  \\ & & traded with the main grid\\
		$\overline{p}^{\mathrm{tr}}$ &[kW]  & max. power traded between prosumers \\
		$Q^{\mathrm{di}}$ &[\euro/kWh$^2$] & quadratic coeff. on the cost of DU \\
		$Q^{\mathrm{st}}$ & [\euro/kWh$^2$]& coeff. on the cost of storage units \\
		$\overline{s}$ & [kVA] & max. line capacity \\
		$T_s$ & [hour] & sampling time \\
		$\overline{v}, \underline{v}$ & p.u. & max. and min. voltage magnitude \\
		$\overline{x},\underline{x}$ &p.u. & max. and min. state of charge \\
		$\overline{\theta}, \underline{\theta}$ & [rad] & max. and min. voltage angle \\
		\hline
		\toprule
		\multicolumn{3}{c}{Sets}  \\
		\hline 
		$\mc B$ & & set of busses in the electrical network \\
		$\mc B^{\mathrm{mg}}$ & & set of busses connected to main grid\\
		$\mathcal{C}$ & & coupling constraint set \\
		$\mc E$ & & set of links in the trading network\\
		$\mathcal{G}^{\mathrm{t}}$ & & graph representing trading network \\
		$\mathcal{G}^{\mathrm{p}}$ & & graph representing physical network \\
		$\mc H$ & & set of discrete-time indices \\
		$\mc L$ & & set of power lines (links)\\
		$\mathcal{N}$ & & set of prosumers \\
		$\mathcal{N}^+$ & & set of prosumers and network operator\\
		$\mathcal{N}_i$ & & set of trading partners of prosumer $i$ \\
		$\mathcal{N}_y^{\mathrm{b}}$ & & set of prosumers of bus $y$ \\
		$\mc P$ & & set of passive consumers \\
		$\mc P_y^{\mathrm{b}}$& &set of passive consumers of bus $y$ \\
		$\mathcal{U}$ & & local constraint set \\ 
		\bottomrule
		
	\end{tabular}
	
\end{table} 


	\section{Introduction}

In recent years, there has been a fast growing penetration of distributed and renewable energy sources as well as storage units in distribution networks \cite{parag2016}. The parties who own these devices are called prosumers, i.e., energy consumers with production and/or storage capabilities. Unlike traditional consumers, prosumers can have a prominent role in achieving energy balance in a distribution network, since they can contribute to energy supply.  Therefore, currently there is a large research effort to study potential evolutions of electricity markets and decentralized energy management mechanisms that can enable active participation of prosumers \cite{parag2016,tushar2018,liu2019,sousa2019}.

Focusing on spot markets, i.e., day-ahead and intra-day markets, each prosumer has to decide its energy production and consumption over a certain time horizon, with the objective of minimizing its own expenses while satisfying its physical and operational constraints. 
Most of existing works formulate such peer-to-peer (P2P) markets via game-theoretic or multi-agent optimization frameworks \cite{tushar2018,tushar2020,lecadre2020,cui2020,zhang2020,baroche2019,sorin2019}. For instance, the authors of \cite{tushar2018} provide a literature survey of early works on game-theoretic P2P market models. {More recently, \cite{tushar2020} considers a coalition game approach for peer-to-peer trading of prosumers with storage units.}  Furthermore, \cite{lecadre2020,cui2020,zhang2020,baroche2019,sorin2019} propose economic dispatch formulations where energy trading is incorporated as coupling (reciprocity) constraints and each prosumer has local decoupled objectives. 

Generalizing the previous papers, our preliminary work in  \cite{belgioioso2020energy} does not only consider multi-bilateral trading  but also trading with the main grid, which extends the coupling to both constraints and objective functions. 
Mathematically, clearing the resulting P2P market corresponds to finding a generalized Nash equilibrium (GNE), namely, a configuration in which no prosumer has an incentive to unilaterally deviate. 
Similarly, \cite{wang2021} formulates a generalized Nash game of energy sharing or a multilateral (instead of bilateral) trading among prosumers, and proposes a distributed algorithm to find a solution of the market equilibrium problem. 
In parallel, we note that operator-theoretic approaches have been effectively exploited to design distributed algorithms that efficiently solve GNE problems under the least restrictive assumptions \cite{paccagnan2019,belgioioso2020b,gadjov2020single,belgioioso2020semi,
	bianchi2020fast}.

In practice, however, direct trading among prosumers might jeopardize system reliability, for which network operators are responsible. Therefore, when designing energy management mechanisms for a distribution grid, one must also consider the role of network operators and the reliability of the system itself.
For example, \cite{qin2018,morstyn2019} treat decentralized markets and operational reliability separately, and propose market-clearing mechanisms where decentralized market solutions must be approved by a network operator based on the system operational constraints. 
An alternative is based on incorporating network charges, which may reflect utilization fees and network congestion, into the market formulation, as discussed in \cite{baroche2019b,paudel2020}. Differently, \cite{moret2020,zhang2020} include network operators as players in the market and impose network operational requirements as constraints in the market problem, which is formulated as a multi-agent optimization. {A similar approach is considered in \cite{zhong2020cooperative}, which employs generalized Nash bargaining theory and decomposes the problem into two hierarchical subproblems (a social welfare maximization and an energy trading problem).} 

In this paper, we consider a P2P energy market in which each prosumer is capable of not only generating and storing energy but also directly trading with other prosumers as well as with the main grid. Similarly to \cite{moret2020}, we include a network operator, whose objective is to ensure safe and reliable operation  of the system. 
However, we formulate the market clearing as a GNE problem, in which the players (i.e., prosumers and network operators) have coupling objective functions and constraints (Section \ref{sub:setup}). {Our market formulation extends the preliminary work \cite{belgioioso2020energy} by including network operational constraints and system operators in the model, which complicate the analysis as we need to exploit the problem structure to derive an efficient algorithm.}

{
The main advantage of our decentralized market design is that its equilibria are not only economically-optimal but also strategically-stable (i.e., no prosumer has any incentive to unilaterally deviate), operationally-safe and reliable (i.e., the network operational requirements are met), and socially-fair (i.e., the marginal loss for satisfying the grid constraints is the same for each prosumer). 
Furthermore, we design a provably-convergent, scalable and distributed market-clearing algorithm based on the proximal-point method for monotone inclusion problems	\cite[\S~23]{bauschke2011convex} (Section \ref{sub:alg}). 
Finally, we investigate via extensive numerical studies: (i)  the effectiveness of the proposed market framework; (ii) the impact of distributed generation, storage and P2P tradings in distribution grids; and (iii) the scalability of the proposed market-clearing mechanism with respect to both the number of prosumers and the number of P2P tradings in the distribution network (Section \ref{sub:case_std}).
}

\smallskip
\subsubsection*{Notation}
$\R$ denotes the set of real numbers, $\bb N$ denotes the set of natural numbers, and  $\bs{0}$ ($\bs{1}$) denotes a matrix/vector with all elements equal to $0$ ($1$). 
$A \otimes B$ denotes the Kronecker product between the matrices $A$ and $B$. For a square matrix $A \in \R^{n \times n}$, its transpose is $A^\top$, $[A]_{i,j}$ represents the element on the row $i$ and column $j$. $A \succ 0$ ($\succcurlyeq 0$) stands for positive definite (semidefinite) matrix. 
{For any $x \in \bb R^n$, $\|x\|_A^2 = x^\top A x,$ with square symmetric matrix $A \succ 0$}. 
For a closed set $S \subseteq \R^n$, the mapping $\proj_{S}:\R^n \rightarrow S$ denotes the projection onto $S$, i.e., $\proj_{S}(x) = \argmin_{y \in S} \left\| y - x\right\|$.

\section{Peer-to-peer markets as a generalized Nash equilibrium problem}
\label{sub:setup}

We denote a group of $N$ prosumers connected in a distribution network by the set $\mc N = \{1,2,\dots, N \}$. Each prosumer might have the capability of producing, storing, and consuming power, depending on their devices and assets. Furthermore, each prosumer might also trade power directly with the main grid and with (some of) the other prosumers, which we will refer to as \textit{trading partners}. The trading partners of an agent might be defined based on geographical location or on bilateral contracts \cite{sousa2019}. 
We model the trading network of prosumers as an undirected graph  $\mc G^{\mathrm{t}} =(\mc N,\mc E)$, where $\mc N $ is the set of vertices (agents) and $\mc E \subseteq \mc N \times \mc N$ is the set of edges, with $|\mc E |= E$. The unordered pair of vertices $(i,j) \in \mc E$ if and only if agents $j$ and $i$ can trade power. The set of trading partners of agent $i$ is defined as $ \mc N_i = \{ j | \, (j,i)\in \mc E \}$.

Moreover, we also consider the electrical distribution network, to which the prosumers are physically connected. This network consists of a set of $B$ busses, denoted by $\mc B:=\{1,2,\dots,B\}$, connected with each other by a set of power lines, denoted by $\mc L\subseteq \mc B \! \times\! \mc B$. Thus, we represent the physical electrical network as a connected undirected graph $\mc G^{\mathrm{p}}=(\mc B, \mc L)$. In $\mc G^{\mathrm{p}}$, each prosumer  is connected to a bus and, in general, one bus may have more than one prosumer. Figure~\ref{fig:GGp} shows an example of trading and physical electrical networks. Furthermore, we assume that a distribution network operator (DNO) is responsible to maintain the reliability of the system, i.e., to ensure the satisfaction of the physical constraints of the electrical network \cite{qin2018,morstyn2019,moret2020}. 
\color{black}

\begin{figure}
	\centering
	\includegraphics[width=0.95\columnwidth]{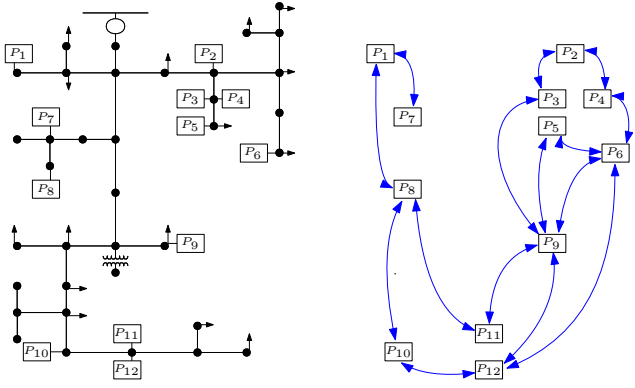}
\smallskip
	\caption{Left plot: A modified IEEE 37-bus network with 12 prosumers (boxes) and 15 passive loads (black triangles); busses are represented by black circles, physical lines in $\mc L$ by solid lines. Right plot: P2P trading network, where trading relations ($\mc E$) are represented by blue double-arrow lines.
	}
	\label{fig:GGp}
\end{figure} 

We focus on P2P spot markets, i.e., day-ahead and intra-day markets, similarly to \cite{sousa2019,lecadre2020,moret2020}. Thus, we denote the horizon of the decision profiles by $\mc H = \{1,2,\dots,H \}$. For instance, in a day-ahead market, typically, the sampling period is one hour and the time horizon is $H=24$ hours. Moreover, as in \cite{moret2020}, we also include the physical constraints of the distribution network to ensure that a solution is not only economically optimal but also meets the standards of the DNO.

Let us model such a P2P market as a generalized game. Specifically, we assume that each prosumer, or agent, $i\in \mc N$ aims at selfishly minimizing its cost function, which might involve decisions of other agents, subject to local and coupling constraints. Furthermore, we consider the DNO as an additional agent, i.e., agent $N\!+\!1$, whose only objective is to ensure the constraints of the physical network are met. 
In this regard, let $u_i \in \bb R^{n_i}$ denote the decision of agent $i$, for all $i\in \mc N^+ :=\{1,\dots,N\!+\!1\}$. 
Furthermore, we denote by $u$ the decision profile, namely, the stacked vectors of the decisions of all agents, i.e., $u:=\col(\{u_j\}_{j\in \mc N^+ })$, and by $u_{-i}$ the decision of all agents except agent $i$, i.e., $u_{-i} = \col(\{u_j\}_{j\in \mc N^+ \backslash \{i\}})$.

Each agent $i$ is self-interested and wants to compute an optimal decision, $u_i^*$, that solves its local optimization problem%
\begin{subequations}
	\label{eq:locOpt}
	\begin{empheq}[left=u_i^*\in \empheqlbrace]{align}
	\label{eq:cost_gen}	
	& 
	\arg\min_{{u}_i }  \;   J_{i}\left(
	u_{i} , u_{-i} 	
	\right) \\[.2em]
	\label{eq:const_gen}
	&\quad \text{s.t. } \; \quad {u}_{i} \in  \mc U_i\\
	&\qquad \quad \; \; (u_i, {u}_{-i}) \in \mc C,
	\end{empheq}
\end{subequations} 
where $J_i$ is the cost function of agent $i$, $\mc U_i$ is the local constraint set, and $\mc C$ is the set of coupling constraints. 
In the remainder of this section, we describe $J_i$, $\mc U_i$, and $\mc C$, upon which we postulate standard assumptions, as formalized next.

\begin{assumption}
	\label{as:gen}
For each agent $i\in \mc N^+$, the function $J_i(\cdot,u_{-i})$ is convex and continuously differentiable, for all fixed $u_{-i}$; the set $\mc U_i$ is nonempty, closed and convex. The global feasible set $\bs {\mc X}:=\left( \prod_{i \in \mc N} \mc U_i \right) \cap \mc C $ satisfies the Slater's constraint qualification \cite[Eq. (27.50)]{bauschke2011convex}.  \eod
\end{assumption}

\subsection{Modelling the prosumers}
\label{sec:mod_comp}
In this section, we introduce the prosumer model. We consider that power might be generated by non-dispatchable generation units, e.g., solar and wind-based generators, or dispatchable units, e.g., small-scale fuel-based generators.
Moreover, we also consider the slow dynamics of storage units. We restrict the model of each component such that Assumption \ref{as:gen} holds, that is, we avoid non-convex formulations and provide a convex approximation instead. Not only this approach is common in the literature, see e.g.   \cite{moret2020,atzeni2013,molzahn2017}, but also practical especially for real-time implementation, which requires fast and reliable computations.
 
First, we suppose that the components of the decision vector of prosumer $i \in \mc N$, $u_{i}$, are the power generated from a \underline{di}spatchable unit ($p_{i}^{\mathrm{di}} \in \bb R^{H}$),	the {\underline{ch}arging and \underline{d}i\underline{s}charging } power of a storage unit {($p_{i}^{\mathrm{ch}}, p_{i}^{\mathrm{ds}} \in \bb R^{H}$)}, the power traded with the \underline{m}ain \underline{g}rid ($p_{i}^{\mathrm{mg}} \in \bb R^H$), and the power \underline{tr}aded with its neighbors $j \in \mc N_i$ ($p_{(i,j)}^{\mathrm{tr}} \in \bb R^H$), for all $j \in \mc N_i$. 
For simplicity of exposition, we assume that each prosumer only owns at most one dispatchable unit and/or one storage unit. 
Next, we present the model for these devices.

\smallskip
\paragraph*{Dispatchable units}
The objective function of a dispatchable unit, denoted by $f_{i}^{\mathrm{di}}: \mathbb{R}^{H} \to \mathbb{R}$, is typically a convex quadratic function   \cite{atzeni2013,hans2018,sorin2019}, e.g., 
\begin{equation}
f_{i}^{\mathrm{di}}(p_{i}^{\mathrm{di}}) = \|p_{i}^{\mathrm{di}}\|_{Q_i^{\mathrm{di}}}^2+ {(c_i^\mathrm{di})}^\top p_{i}^{\mathrm{di}}, \label{eq:f_dg}
\end{equation}
where $Q_i^{\mathrm{di}}\succcurlyeq0$ and $c_i^{\mathrm{di}}$ are constants.
Furthermore, the power generation $p_{i}^{\mathrm{di}}$ is limited by 
\begin{equation}
\begin{aligned}
\underline{p}_{i}^{\mathrm{di}}\1_{H} \leq p_{i}^{\mathrm{di}} &\leq \overline{p}_{i}^{\mathrm{di}}\1_{H}, \quad& &\forall i\in\mathcal{N}^{\mathrm{di}},\\
p_{i}^{\mathrm{di}} &= 0, \quad& &\forall i\notin\mathcal{N}^{\mathrm{di}},
\end{aligned}
\label{eq:p_dg_bound}	
\end{equation}
where $\overline{p}_{i}^{\mathrm{di}} > \underline{p}_{i}^{\mathrm{di}}\ge 0$ denote maximum and minimum total power production of the dispatchable generation unit, and $\mathcal{N}^{\mathrm{di}}\subseteq\mathcal{N}$ the subset of agents that own dispatchable units.

\paragraph*{Storage units} 
Each prosumer might also minimize the usage of its storage units, for instance, in order to reduce its degradation. The corresponding cost function is denoted by $f_{i}^{\mathrm{st}}:\mathbb{R}^{H} \to \mathbb{R}$ and defined as in \cite{hans2018} as follows:
\begin{equation}
f_{i}^{\mathrm{st}}(p_{i}^{\mathrm{ch}}, p_{i}^{\mathrm{ds}})= \|p_{i}^{\mathrm{ch}}\|_{Q_i^{\mathrm{st}}}^2 + \|p_{i}^{\mathrm{ds}}\|_{Q_i^{\mathrm{st}}}^2, \label{eq:f_st}
\end{equation}
where $Q_i^{\mathrm{st}}\succcurlyeq0$. {The battery charging and discharging profiles, $p_{i}^{\mathrm{ch}} = \col((p_{i,h}^{\mathrm{ch}})_{h\in \mc H})$ and $p_{i}^{\mathrm{ds}} = \col((p_{i,h}^{\mathrm{ds}})_{h\in \mc H})$, respectively, are constrained by the battery dynamics \cite{atzeni2013,zhong2019online}},
\begin{equation}
	\begin{aligned}
			x_{i,h+1} &=\eta_i^{\mathrm{st}} x_{i,h} +\tfrac{T_{\mathrm{s}}}{e^{\mathrm{cap}}_{i}}(\eta_i^{\mathrm{ch}} p_{i,h}^{\mathrm{ch}} -(\tfrac{1}{\eta_i^{\mathrm{ds}}})p_{i,h}^{\mathrm{ds}}), 
			\underline{x}_i &\leq x_{i,h+1} \leq \overline{x}_i,  \qquad \forall i \in \mc N^{\mathrm{st}}, \forall h \in \mc H,\\
			{p}_{i}^{\mathrm{ch}} &\in [0,\overline{{p}}_{i}^{\mathrm{ch}}], \ \ {p}_{i}^{\mathrm{ds}} \in [0,\overline{{p}}_{i}^{\mathrm{ds}}],  \quad  \forall i \in \mc N^{\mathrm{st}},
	\\
		p_{i}^{\mathrm{ch}} &= 0, \ \ p_{i}^{\mathrm{ds}}= 0,\qquad \qquad \quad \  \forall i \notin \mc N^{\mathrm{st}},
	\end{aligned}
	\label{eq:x}
\end{equation}
where $x_{i,h} $ denotes the state of charge (SoC) of the storage unit at time $h \in \mc H$, 
{$\eta_i^{\mathrm{st}}, \eta_i^{\mathrm{ch}}, \eta_i^{\mathrm{ds}} \in (0,1]$ denote  the leakage coefficient of the storage, charging, and discharging efficiencies, respectively, while }  
$T_{\mathrm{s}}$ and $e^{\mathrm{cap}}_{i}$ denote sampling time and maximum capacity of the storage, respectively. Moreover, $\underline{x}_i,\overline{x}_i  \in [0,1]$  denote the minimum and the maximum SoC of the storage unit of prosumer $i$, respectively, whereas ${\overline p^{\mathrm{ch}}_i} \geq 0$ and ${\overline p^{\mathrm{ds}}_i} \geq 0$ denote the maximum charging and discharging power of the storage unit. Finally, we denote by  $\mathcal{N}^{\mathrm{st}}\subseteq\mathcal{N}$ 	the set of prosumers that own a storage unit.

\smallskip
\paragraph*{Local power balance} 
The local power balance of each prosumer $i \in \mc N$ is represented by the following equation:
\begin{equation}
p_{i}^{\mathrm{di}}+{p_{i}^{\mathrm{ds}}-p_{i}^{\mathrm{ch}}}+p_{i}^{\mathrm{mg}}+\sum_{j\in \mathcal{N}_i}p_{(i,j)}^{\text{tr}}=p_{i}^{\mathrm{d}}, \label{eq:l_pow_b}
\end{equation}
where $p_{i}^{\mathrm{d}} \in \mathbb{R}^{H}$ denotes the local power demand profile over the whole prediction horizon. The power demand $p_{i}^{\mathrm{d}}$ is defined as the difference between the aggregate load of prosumer $i$ and the power generated by its non-dispatchable generation units, e.g., solar or wind-based generators\footnote{If a component of $p_{i}^{\mathrm{d}}$ is positive, then the load is larger than the power produced by its non-dispatchable units.}.
Finally, it is worth mentioning that a prosumer that does not own a dispatchable nor storage unit can satisfy its power balance \eqref{eq:l_pow_b} by importing (trading) power from other prosumers and/or the main grid.

\smallskip
\paragraph*{Passive consumers}
In addition, we assume that some busses in the distribution network might also be connected to some (traditional) passive consumers that do not have storage nor dispatchable units, and do not trade with other prosumers. Let us denote the set of such passive consumers by $\mc P$. For each passive consumer $i \in \mc P$, its power demand $p_i^{\mathrm d}>0$ is balanced conventionally, namely, by importing power from the main grid. Nevertheless, these passive loads will play a role in the trading process between prosumers and main grid, and in the power-balance equations of the physical network.
\color{black}

\subsection{Modelling the P2P trading}
\label{sec:mod_tr}
 
In this section, we present the cost and constraints of bilateral tradings between prosumers.

\smallskip
\paragraph*{Power traded with neighbors}
Recall that each prosumer $i \in \mc N$ has a set of trading partners denoted by $\mc N_i$. The corresponding cumulative trading cost is given by
\begin{equation}
f_{i}^{\mathrm{tr}} \left( \{ p_{(i,j)}^{\mathrm{tr}} \}_{j \in \mc N_i} \right) = 
\1_H^\top \sum_{j\in \mathcal{N}_i} \left(
c_{(i,j)}^{\mathrm{tr}}p_{(i,j)}^{\mathrm{tr}}\!+\! c^{\mathrm{ta}}  |p_{(i,j)}^{\mathrm{tr}}|\right)
, \label{eq:f_t}
\end{equation}%
where $p_{(i,j)}^{\mathrm{tr}} \in \R^{H}$ is the power that prosumer $i$ trades with prosumer $j$, $c_{(i,j)}^{\mathrm{tr}} \geq 0$ is the per-unit cost of trading \cite{lecadre2020}, and $c^{\mathrm{ta}}$ is a tariff imposed by the DNO for using the network \cite{baroche2019}. In  practice, the parameters $c_{(i,j)}^{\mathrm{tr}}$ can be agreed through a bilateral contract \cite{sousa2019}, model taxes to encourage the development of certain technologies  {or be used for the purpose of product differentiation} \cite{sorin2019,baroche2019,lecadre2020}. Furthermore, for each P2P trade it must hold that
\begin{subequations}
\begin{align}
-\overline{p}_{(i,j)}^{\mathrm{tr}}\1_H \leq p_{(i,j)}^{\mathrm{tr}} &\leq \overline{p}_{(i,j)}^{\mathrm{tr}}\1_H, & & \forall j\in\mathcal{N}_i, \label{eq:p_t_bound}\\
p_{(i,j)}^{\mathrm{tr}}  + p_{(j,i)}^{\mathrm{tr}}&=0, & & \forall j\in\mathcal{N}_i, \label{eq:rep_const}
\end{align}
\label{eq:ptr_cons}%
\end{subequations}
where $ \overline{p}_{(i,j)}^{\mathrm{tr}}$ denotes the maximum power can be traded with neighbor $j$. Equations \eqref{eq:rep_const}, commonly known as \textit{reciprocity constraints} \cite{sousa2019}, impose the agreement on the power trades.

\paragraph*{Power traded with the main grid} Let $p_{i,h}^{\mathrm{mg}}$ be the power prosumer $i$ imports from the main grid at time $h \in \mc H$. As in \cite{atzeni2013}, we assume that the electricity unit price at each time step $h \in \mc H	$ depends on the total consumption,
\begin{equation}
c_h^{\mathrm{mg}}(\sigma^{\textrm{mg}}_h)= d_h^{\mathrm{mg}}\cdot {\left(\sigma^{\textrm{mg}}_h + b_h \right)}^2,
\end{equation}
where $d_h^{\mathrm{mg}}$ is a positive price parameter, whereas $\sigma^{\textrm{mg}}_h$ and $b_h$ denote the \textit{aggregate active} and \textit{passive} \textit{load on the grid}, i.e., 
\begin{equation} \label{eq:Aggr}
\sigma^{\textrm{mg}}_h = \sum_{i\in\mathcal{N}} p_{i,h}^{\mathrm{mg}},
\quad
b_h =\sum_{i\in \mathcal P}p_{i,h}^{\mathrm{d}},
 \quad \forall h \in \mathcal H.
\end{equation} 
Therefore, the total cost incurred by prosumer $i$, over the horizon $\mc H$, for trading with the main grid is given by
\begin{equation}
\begin{aligned} \textstyle
f_{i}^{\mathrm{mg}}\left( p_{i}^{\mathrm{mg}},\sigma^{\textrm{mg}} \right) &= \sum_{h\in \mc H}c_h^{\mathrm{mg}}(\sigma^{\textrm{mg}}_h) \, \frac{p_{i,h}^{\mathrm{mg}}}{\sigma_h^{\textrm{mg}} + b_h}\\
&= \sum_{h\in \mc H}  d_h^{\mathrm{mg}}(\sigma^{\textrm{mg}}_h+b_h) \, p_{i,h}^{\mathrm{mg}},
\label{eq:f_mg}
\end{aligned}
\end{equation}
We note that the cost function \eqref{eq:f_mg} assumes equal electricity price at each distribution node and the consideration of power losses and congestion, which may result in different price at different node, is left for future work.  

Finally, we bound the aggregative loads \eqref{eq:Aggr} as follows:
\begin{align}
\underline{p}^{\mathrm{mg}} \1_H \leq \sigma^{\textrm{mg}} + b  \leq \overline{p}^{\mathrm{mg}} \1_H   , \label{eq:p_mg_bound}
\end{align}
where $\overline{p}^{\mathrm{mg}}>\underline{p}^{\mathrm{mg}} \geq 0$ denote the upper and lower bounds. Typically, the latter is  positive to ensure a continuous operation of the main generators that supply the main grid.

\subsection{Modelling the physical constraints}
To ensure that the solutions to our decentralized market design are operationally-safe and reliable for the entire system, we impose the physical constraints of the electrical network, namely, power-flow-related constraints.

Firstly, recall that $\mc G^{\mathrm p}=(\mc B, \mc L)$ is a graph representation of the physical electrical network that connects the prosumers. We  denote by $\mc B_y =\{z~|~(y,z) \in \mc L\}$ the set of neighbouring busses of bus $y \in \mc B$, whereas we denote by $\mc N_y^{\mathrm b} \subseteq \mc N$ and $\mc P_y^\mathrm{b} \subseteq \mc P$ the set of prosumers and passive consumers that are connected to bus $y \in \mc B$, respectively. 
Additionally, we denote the set of busses connected to the main grid by $\mc B^{\mathrm{mg}} \subseteq \mc B$.

Secondly, we define decision variables, for each bus $y \in \mc B$, which are used to define the physical constraints. Denote by $v_y \in \bb R^H$ and $\theta_y \in \bb R^H$ the voltage magnitude and angle over $\mc H$.  Moreover, $p^{\mathrm{tg}}_y \in \mathbb{R}^H$ denotes the real power exchanged between bus $y \in \mc B$ and the main grid, whereas $p_{(y,z)}^{\ell}$ and $q_{(y,z)}^{\ell} \in \bb R^H$, for each $m \in \mc B_y$, denote the real and reactive powers of line $(y,z) \in \mc L$ over $\mc H$, respectively. 

\smallskip
We consider a linear approximation of power-flow equations, which is standard in the literature of P2P markets, e.g., \cite{yang2019,moret2020}. Specifically, for each bus $y \in \mc B$, it must hold that
\begin{equation}
\sum_{i \in \mc P_y^\mathrm{b}}p_i^{\mathrm{d}}+ \sum_{i\in \mc N_y^{\mathrm b}} \eta_i-p_{y}^{\mathrm{tg}}=\sum_{z\in \mc B_y}p_{(y,z)}^{\ell}, \label{eq:p_bal_p}
\end{equation}
where $\eta_i$ is the active power injection of prosumer $i$, i.e.,
\begin{equation}\label{eq:etai}
\eta_i:= p_i^{\mathrm d}-p_{i}^{\mathrm{di}}-{p_{i}^{\mathrm{ds}}+p_{i}^{\mathrm{ch}}}.
\end{equation}
Equation \eqref{eq:p_bal_p} models the local power balance of bus $y$, similarly to \eqref{eq:l_pow_b} although now it relates power generation, consumption, and line powers. Moreover, it must hold that\begin{subequations}
\small
\label{eq:pf_pq}
\begin{align}
&p_{(y,z)}^{\ell} = B_{(y,z)}\left({\theta_y - \theta_z} \right) - G_{(y,z)}\left(v_y - v_z \right), \quad \forall z \in \mc B_y, \label{eq:pf_p}\\
&q_{(y,z)}^{\ell} = G_{(y,z)}\left({\theta_y - \theta_z} \right) + B_{(y,z)}\left(v_y - v_z \right), \quad  \forall z \in \mc B_y, \label{eq:pf_q}
\end{align}
\label{eq:pf_all}%
\end{subequations}
which represent the power flow equations of line $(y,z)$ from the perspective of bus $y$, with $B_{(y,z)}$ and  $G_{(y,z)}$ denoting the susceptance and conductance, respectively, of line $(y,z)$. Note that by \eqref{eq:pf_p} and \eqref{eq:pf_q}, for each pair $(y,z) \in \mc L$, it holds that $p_{(y,z)}^{\ell} = {-p}_{(z,y)}^{\ell}$ and $q_{(y,z)}^{\ell} = {-q}_{(z,y)}^{\ell}$. 

\smallskip
We also impose reliability constraints for each bus $y \in \mc B$, 
\begin{subequations}
\label{eq:line}
\begin{align}
	(p_{(y,z),h}^{\ell})^2 + (q_{(y,z),h}^{\ell})^2 &\leq \overline{s}_{(y,z)}^2,\quad \forall z\in\mc B_y,  \forall h\in \mc H, \label{eq:line_cap}\\
	\underline{\theta}_y\1 \leq \theta_y &\leq \overline{\theta}_y\1,  \label{eq:theta_b}\\
	\underline{v}_y\1 \leq v_y &\leq \overline v_y\1, \label{eq:v_b}
\end{align}
\end{subequations}
where \eqref{eq:line_cap} represents the line capacity constraint at each line, with maximum capacity of line $(y,z) \in \mc L$ denoted by $\overline{s}_{(y,z)}$, and \eqref{eq:theta_b}-\eqref{eq:v_b} represent the bounds of the voltage phase angles and magnitudes, respectively, with $\underline{\theta}_y \leq \overline{\theta}_y$ denoting the minimum and maximum phase angles and $\underline{v}_y\leq \overline v_y$ denoting the minimum and maximum voltage magnitude. Note that, when linearizing the power flow equations, we take one of the busses as reference bus. Without loss of generality, we suppose the reference is bus $1$ and assume $\underline{\theta}_1 = \overline{\theta}_1=0$.  \color{black}

Finally, the power exchanged with the main grid must satisfy the following constraints:
\begin{subequations}
\begin{align}
	p^{\mathrm{tg}}_{y} &= 0,  &\forall y \notin \mc B^{\mathrm{mg}}, \label{eq:grid_ex}\\
	\sigma^{\textrm{mg}}_h + b_h &= \sum_{y\in\mc B}p^{\mathrm{tg}}_{y,h},  &\forall h \in \mathcal H, \label{eq:grid_const2}
\end{align}
\end{subequations}
where \eqref{eq:grid_ex} is imposed by definition that the busses that are not directly connected with the main grid do not exchange  power with the main grid, whereas \eqref{eq:grid_const2} ensures that the power traded by the prosumers with the main grid (in the trading network) corresponds to the power exchanged between the whole distribution network and the main grid.

\section{A Distributed Market-clearing Mechanism}
\label{sub:alg}
\subsection{Market-Clearing Game and Variational Equilibria}
By letting the physical variables of the distribution network be handled by a DNO (i.e., agent $N\!+\!1$), the P2P market clearing problem can be compactly written as the problem of finding the optimal strategy profiles $u_i^*$'s in \eqref{eq:locOpt}, for all $i \in \mc N^+$, where the decision variable $u_i$ is defined as
$$u_i \hspace{-2pt} = \hspace{-2pt}\begin{cases} \col\hspace{-2pt}\left(p_{i}^{\mathrm{di}},{p_{i}^{\mathrm{ch}},p_{i}^{\mathrm{ds}}, } p_{i}^{\mathrm{mg}},\{p_{(i,j)}^{\mathrm {tr}}\}_{j\in \mc N_i}\right)\hspace{-2pt}, \ \qquad \qquad \;  \forall  i \in \mc N,\\
	\col\hspace{-2pt}\left(\{\theta_y, v_y, p_y^{\mathrm{tg}},\{p_{(y,z)}^{\ell},q_{(y,z)}^{\ell}\}_{z \in \mc B_y} \}_{y\in\mc B} \right)\hspace{-2pt},  \; \quad i\!=\! N\!+\!1;
		\end{cases}$$ 
the cost function is defined as
	\begin{multline}
	J_i(u_i,u_{-i}) =
f_{i}^{\mathrm{di}}(p_{i}^{\mathrm{di}}) + f_{i}^{\mathrm{st}}({p_{i}^{\mathrm{ch}},p_{i}^{\mathrm{ds}}})+ f_{i}^{\mathrm{tr}} \left( \{ p_{(i,j)}^{\mathrm{tr}} \}_{j \in \mc N_i} \right)\\
    + f_{i}^{\mathrm{mg}}\left( p_{i}^{\mathrm{mg}},\sigma^{\textrm{mg}} \right), \quad \forall  i \in \mc N,
 \label{eq:cost_def}
	\end{multline}
whereas\footnote{Here, we assume that the DNO does not have preferences on the outcome, provided that it is a feasible solution for the grid.} $J_{N+1}=0$; the local action set is
\begin{align}
\label{eq:constr_def}
\mc U_i =
\left\{
\begin{array}{l r}
\left\{u_i \,|~ \text{\eqref{eq:p_dg_bound}, \eqref{eq:x}, \eqref{eq:l_pow_b}, \eqref{eq:p_t_bound} hold} \, \right\}, & \forall i \in \mc N, \\
\left\{ u_i \,|~ \text{\eqref{eq:pf_pq},\eqref{eq:line}, \eqref{eq:grid_ex} hold} \,\right\}, & i\!=\!N\!+\!1;
\end{array}
\right.
\end{align}
and finally, the set of coupling constraints is 
\begin{equation}
\label{eq:Cconstr_def}
\mc C = \left\{
u~|~ \text{\eqref{eq:rep_const}, \eqref{eq:p_mg_bound}, \eqref{eq:p_bal_p}, \eqref{eq:grid_const2} hold}
\right\}.
\end{equation}

\begin{remark}
The definitions of $J_i$, $\mc U_i$, $\mc C$ in \eqref{eq:cost_def}, \eqref{eq:constr_def}, \eqref{eq:Cconstr_def} satisfy Assumption~\ref{as:gen}. Moreover, these definitions can be expanded by incorporating additional cost terms, for example, related to the degradation of storage units and constraints (e.g. ramping constraints of dispatchable generation units), as long as Assumption \ref{as:gen} remains satisfied.  {Additionally, instead of linear power flow equations in \eqref{eq:pf_all}, a nonlinear convex relaxation, such as a second order cone or semi-definite programming as discussed in \cite[Sect. II-A]{molzahn2017} can be considered since it still satisfies Assumption 1. In this case, only the definition of $\mc U_{N+1}$ differs from the current formulation.}\eod
\end{remark}

From a game-theoretic perspective, the collection of inter-dependent optimization problems in \eqref{eq:locOpt} constitute a \textit{generalized game}, and a set of decisions $\{u_1^*,\ldots,u_{N+1}^*\}$ that simultaneously satisfy \eqref{eq:locOpt}, for all $i \in \mc N^+$, corresponds to a GNE \cite[\S~2]{facchinei201012}.
In other words, a set of strategies $\{u_1^*,\ldots,u_{N+1}^*\}$ is a GNE if no agent $i \in \mc N^+$ (prosumers and DNO)  can reduce its cost  function $J_{i}(u_{i}^*, u_{-i}^*)$ by unilaterally changing its strategy $u_i^*$ to another feasible one.
{Among all GNEs, we target the special subclass of \textit{variational} GNEs (v-GNEs) that coincides with the solutions to a specific variational inequality GVI$(\mc K, P)$ \cite[Prop.~12.4]{facchinei201012},  i.e., the problem of finding a pair of vectors $(u^*,z^*)$, such that $u^* \in \mc K$, $z^* \in P(u^*)$, and
\begin{align*}
{(u - u^*)}^\top z^* \geq 0, \quad \forall u  \in  {\mc K}, 
\end{align*}
where the mapping $P(u^*): = \prod_{i \in \mc N^+} \frac{\partial}{\partial u_i} J_i(u_i^*,u_{-i}^*)$ is the so-called \textit{pseudo-subdifferential}, and $\mc K:= \mc C \cap (\prod_{i \in \mc N^+} \mc U_i)$ is the global feasible set.
v-GNEs enjoy the property of ``economic fairness”, namely, the marginal loss due to the presence of the coupling constraints is the same for each agent, see e.g. \cite{kulkarni2012variational}. For these reasons, v-GNEs have been used to model desirable (i.e., efficient, strategically stable, fair, and safe) configurations in several distributed engineering systems, including P2P energy market models, see e.g. \cite{lecadre2020}. In this paper, we focus on computational aspects, namely, the design and analysis of a fast and scalable decentralized v-GNE seeking algorithm for the P2P market game \eqref{eq:locOpt}, while we study the properties of its v-GNEs numerically rather than analytically.}

Note that the cost functions in \eqref{eq:cost_def} are coupled only via the aggregative quantity $\sigma^{\textrm{mg}} = \sum_{i\in\mathcal{N}} p_{i}^{\mathrm{mg}}$ in  \eqref{eq:Aggr}, namely, the active load (i.e., the congestion) on the main grid. Therefore, for each agent $i \in \mc I$, we can define a function $\tilde J_i$ such that
\begin{equation}
\tilde J_i(u_i, \sigma^{\textrm{mg}}) :=J_i(u_i,  u_{-i}).
\end{equation}
Games with such special structure are known as \textit{aggregative games} \cite{jensen2010aggregative}, and have received intense research interest, within the operations research and the automatic control communities \cite{paccagnan2019,belgioioso2020semi,belgioioso2020b,
gadjov2020single,bianchi2020fast}. 
{
When the agents' cost functions depend linearly on the congestion (as for our P2P market model) v-GNEs are efficient (in terms of social welfare). Specifically, the so-called \textit{price of anarchy}\cite{koutsoupias1999worst}, which quantifies how much selfish behaviour degrades the performance of a given system, tends to one (i.e., no performance degradation) as the agents population size grows unbounded \cite{paccagnan2018efficiency}.}

\subsection{Semi-decentralized Market Clearing}
Several semi-decentralized and distributed algorithms have been recently proposed to find a solution of the generalized aggregative game in \eqref{eq:locOpt}, e.g. \cite{paccagnan2019,belgioioso2020semi,belgioioso2020b,
gadjov2020single,bianchi2020fast}. 
{%
Among these methods, we focus on semi-decentralized ones \cite{belgioioso2020semi}, in which the agents (i.e., prosumers) rely on a reliable central coordinator (i.e., the DNO) that gathers local variables in aggregative form and then broadcasts (incentive) signals for coordination purposes. 

In this paper, we exploit the special linear coupling structure in the cost functions \eqref{eq:cost_def} and coupling constraints \eqref{eq:Cconstr_def} to tailor Algorithm~6 in \cite{belgioioso2020semi} for our P2P market game.
Unlike most of the available semi-decentralized pseudo-gradient-based methods, \cite[Algorithm~6]{belgioioso2020semi} relies on proximal updates that are computationally more expensive but greatly mitigate the overall communication burden between agents and coordinator.
The resulting market-clearing mechanism, summarized in Algorithm~1, requires the prosumers and the DNO to store, update, and communicate some additional (dual and auxiliary) variables, whose primary function is to coordinate the system towards operational feasibility and trading reciprocity. In particular, each prosumer $i\in\mc N$ stores in its local memory
\begin{itemize}
\item the local strategy $u_i$ that collects the power generation, storage (charging/discharging), load, and trading profiles;

\item the active power injection $\eta_i$, defined as in \eqref{eq:etai} and privately communicated (as a grid usage bid) to the DNO;

\item a (dual) variable $\mu^{\text{tr}}_{(i,j)}$ for each trading partner $j \in \mc N_i$, whose function is to drive prosumer $i$'s and $j$'s power trades to agreement, i.e., the reciprocity constraints \eqref{eq:rep_const}, and can be interpreted from an economic perspective as a bilateral trading shadow price \cite[Section 2.4]{lecadre2020}.

\end{itemize}
In addition to the physical variables of the distribution network, i.e., $u_{N\!+\!1}$, the DNO stores in its local memory
\begin{itemize}
\item the (dual) variables $\lambda^{\text{mg}}$ and $\mu^{\text{tg}} $, that are associated with the main grid constraints \eqref{eq:p_mg_bound} and \eqref{eq:grid_const2}, respectively;

\item a (dual) variable $\mu_{y}^{\text{pb}}$ for each bus $y \in \mc B$, associated with the power balance constraint on bus $y$ \eqref{eq:p_bal_p}.
\end{itemize}
From an economic perspective, these variables can be interpreted as extra marginal losses imposed to the prosumers for the grid usage. From a control-theoretic perspective, they can be interpreted as states of discrete-time integrators driven by the violation of the network operational constraints \eqref{eq:Cconstr_def}.
}
\begin{figure}[t!]%
\begin{minipage}{\columnwidth}
\hrule
\smallskip
\textsc{Algorithm 1}. Semi-decentralized P2P Market Clearing
\smallskip
\hrule 
\smallskip
\text{Initialization}: For all prosumers $i \in \mc N$: set $\mu^{\text{tr}}_{(i,j)}(-1) = 0$, $\forall j \in \mc N_i$. DNO: set $\lambda^{\text{mg}}(0)\!=\!0$, $\mu^{\text{tg}}(0) \!=\! 0$, $\mu_y^{\text{pb}}(0)\! = \!0$, $\forall y \!\in\! \mc B$.

\smallskip
 \text{Iterate until convergence $(k = 0,1,\ldots)$}
\begin{align*}
& \text{For all prosumers $i \in \mc N$ (in parallel):}\\
\quad &\left\lfloor
\begin{array}{l}
\text{Local update via \textsc{Algorithm 2}:}\\
\left\lfloor
\begin{array}{l}
\text{Set $\{\mu_{(i,j)}(k)\}_{j \in \mc N_i}$ as in \textsc{Alg. 2} (i)}\\
\text{Set $u_i(k\!+\!1)$ as in \textsc{Alg. 2} (ii)}
\end{array}
\right.\\[1em]
 \text{Communication}\\
 \left\lfloor
\begin{array}{l}
\eta_i(k\!+\!1), \; p_i^{\text{mg}}(k\!+\!1) \rightarrow \text{ DNO}\\[.5em]
\text{For all trading partners } j \in \mc N_i \text{ (in parallel):}\\
 \left\lfloor
\begin{array}{l}
p_{(i,j)}^{\text{tr}}(k\!+\!1) \rightarrow \text{ prosumer } j
\end{array}
\right.
\end{array}
\right.
\vspace*{.1em}
\end{array}
\right.\\[.5em]
& \text{Distribution Network Operator (DNO)}\\
\quad &\left\lfloor
\begin{array}{l}
\text{Central update via \textsc{Algorithm 3}:}\\
\left\lfloor
\begin{array}{l}
\text{Set $u_{N+1}(k\!+\!1)$ as in \textsc{Alg. 3} (i)}\\
\text{Set $\lambda^{\text{mg}}$, $\mu^{\text{tg}}$, $\{ \mu_y^{\text{pb}}(k\!+\!1)\}_{y \in\mc B}$ as in \textsc{Alg. 3} (ii)}
\end{array}
\right.\\[1em]
\text{Communication (Broadcast)}\\
 \left\lfloor
\begin{array}{l}
\sigma^{\mathrm{mg}},\, \lambda^{\text{mg}}, \mu^{\text{tg}}(k\!+\!1) \rightarrow \text{ all prosumers } i \in \mc N\\[.5em]
\text{For all busses } y \in \mc B \text{ (in parallel):}\\
 \left\lfloor
\begin{array}{l}
\mu_y^{\text{pb}}(k+1) \rightarrow \text{ all prosumers } i \in \mc N_y^{\text{b}} \text{ on bus } y%
\end{array}%
\right.%
\end{array}%
\right.
\vspace*{.1em}
\end{array}
\right.%
\end{align*}
\smallskip
\hrule
\end{minipage}

\vspace*{2.4em}

\centering
\includegraphics[width=.36\textwidth]{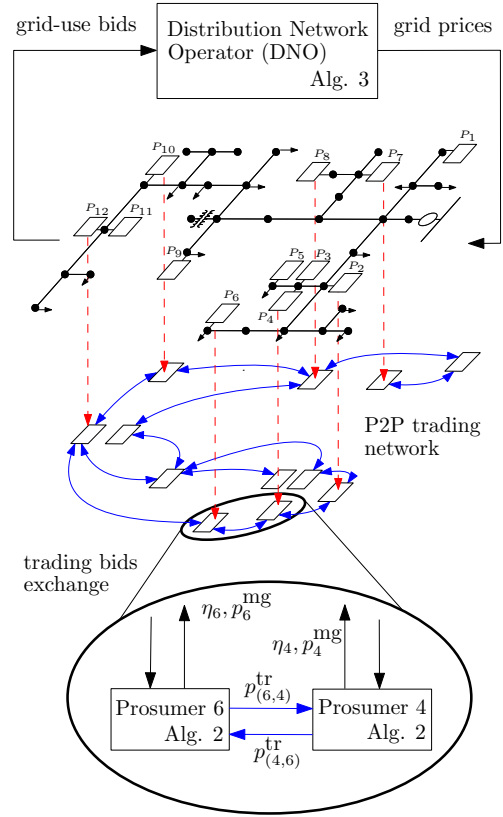}
\smallskip
\caption{Information flow in Algorithm 1.}
\label{fig:IFC}
\end{figure}%

The semi-decentralized information flow of Algorithm~1 is illustrated in Figure~\ref{fig:IFC}, while its locals and central updates are summarized  in Algorithm~2 and Algorithm~3, respectively. In there, we used some auxiliary variables (e.g., $\zeta^{\text{tr}}_{(i,j)}$, $\psi_i$ in Algorithm~2) to keep the presentation compact.  
The next proposition shows the global convergence of Algorithm 1 to a variational GNE of the proposed P2P market game. {Due to space limitations, we provide only a sketch of the proof that is mainly based on the technical results in \cite[Theorem~2]{belgioioso2020semi}.}

\smallskip
\begin{proposition} The following statements hold true:\hspace*{3em}
	\label{prp:conv}
\begin{enumerate}[(i)]
\item There exists a v-GNE of the P2P market game \eqref{eq:locOpt}.
\item The sequence $\left(u_1(k),\ldots,u_{N+1}(k)\right)_{k \in \bb N}$ generated by Algorithm~1 converges to a v-GNE of \eqref{eq:locOpt}.
\end{enumerate}
\end{proposition}
\begin{proof} 
See Appendix \ref{app:CA}.
\end{proof}

\begin{figure}[t]
\begin{minipage}{\columnwidth}
\hrule
\smallskip
\textsc{Algorithm 2}. Local update of Prosumer $i$
\smallskip
\hrule 
\smallskip
\text{Step sizes}: For each $i \in \mc N$, set $\alpha_{i} <1/(3+N\max_{h\in \mc H} d_h^{\mathrm{mg}})$, $\beta_{(i,j)}^{\text{tr}}=\beta_{(j,i)}^{\text{tr}} < 1/2$, for all $j \in \mc N_i$.

\medskip
\noindent 
(i) \text{Dual update (trading reciprocity):}\\[.2em]
\hspace*{.15em}
$
\left\lfloor
\begin{array}{l}
\text{For all } j \in \mc N_i \text{ (in parallel):}\\
\left\lfloor
\begin{array}{l}
\zeta^{\text{tr}}_{(i,j)}(k)  =  p^{\textrm{tr}}_{(i,j)} (k) + p^{\textrm{tr}}_{(j,i)} (k)\\
\mu^{\text{tr}}_{(i,j)}(k) = \mu^{\text{tr}}_{(i,j)}(k) + \beta_{(i,j)}^{\text{tr}} \left( 
	2  \zeta^{\text{tr}}_{(i,j)}(k) - \zeta^{\text{tr}}_{(i,j)}(k\!-\!1)
	\right)
\end{array}
\right.
\vspace*{.1em}
\end{array}
\right.
$

\medskip
\noindent
(ii) \text{Primal update (generation, storage, load, and trading):}\\[.2em]
\hspace*{.15em}
$
\left\lfloor
\begin{array}{l}
\psi_i(k) = u_i(k)-
	\alpha_i \cdot \col \Big(-\mu_y^{\text{pb}}(k),\mu_y^{\text{pb}}(k), -\mu_y^{\text{pb}}(k),\\
	\hspace*{5em}
	{
	\left[
	\begin{smallmatrix}
	I_H\\
	- I_H 
	\end{smallmatrix}
	\right] 
	}^\top \lambda^{\text{mg}}(k) + \mu^{\text{tg}}(k),
	\left\{ {\mu^{\text{tr}}_{(i,j)}(k)}\right\}_{j \in \mc N_i}   \Big)\\
\text{Set } u_i(k\!+\!1) \text{ as the unique solution to}\\[.2em]	
\left\{
\begin{array}{r l}
	\underset{\xi \in \R^{n_i}}{\argmin} & 
	 J_{i} \big( \xi, u_{-i}(k) \big) 
	+ \frac{1}{2 \alpha_i}
	\left\| \xi -  \psi_i(k) 
	\textstyle
	  \right\|^2 \\
	\text{s.t.} & \xi \in \mc U_i
\end{array} 
	\right.
	\vspace*{.1em}
\end{array}	
\right.		$

\medskip
\hrule
\end{minipage}
\end{figure}
\begin{figure}[t]
\begin{minipage}{\columnwidth}
\hrule
\smallskip
\textsc{Algorithm 3}. DNO central update
\smallskip
\hrule 
\smallskip

\text{Step sizes}: set $\alpha_{N+1} <2$, $\gamma^{\text{mg}} < 1/N$, $\beta^{\text{tg}} < (|\mc N| \!+\! |\mc B|)^{-1}$, and $\beta_y^{\text{pb}} < (1\!+\!2|\mc N_y^{\mathrm{b}}|\!+\!|\mc B_y|)^{-1}$,  for all busses $y\!\in\! \mc B $.

\medskip
\noindent 
(i) \text{Primal update (grid physical variables):}\\[.2em]
\hspace*{.15em}
$\left\lfloor
\begin{array}{l}
\psi(k) = \col \left( 
	\left\{
	\0,\mu^{\text{tg}}(k) + \mu_y^{\text{pb}}(k), 
	\{ \mu_y^{\text{pb}}(k), \0 \}_{z \in \mc B_y}
	\right\}_{y \in \mc B} \right)\\[.7em]
u_{N+1}(k+1) = \proj_{\mc U_{N+1}}  \left( u_{N+1}(k) + (\alpha_{N+1})^{-1} \psi(k) \right)
\vspace*{.1em}
\end{array}
\right.
$

\medskip
\noindent
\noindent 
(ii) \text{Dual update (operational feasibility):}\\[.2em]
\hspace*{.15em}
$
\left\lfloor
\begin{array}{l}
\delta(k\!+\!1) = 
	\left[
	\begin{smallmatrix}
	1\\
	- 1 
	\end{smallmatrix}
	\right] \otimes (2 \sigma^{\text{mg}}(k\!+\!1)\!-\! \sigma^{\text{mg}}(k)) 
	-
	\left[
	\begin{smallmatrix}
	\overline{p}^{\mathrm{mg}}\1_{H}-b \\
	-     \underline{p}^{\mathrm{mg}} \1_{H}+b
	\end{smallmatrix} 
	\right]\\
	\lambda^{\text{mg}}(k+1) = \textstyle
	\proj_{\R^{2 H}_{\geq 0}}\left( 
	\lambda^{\text{mg}}(k) + \gamma^{\text{mg}} \delta(k\!+\!1)
	 \right)\\
\zeta^{\text{tg}}(k+1) = \sigma^{\text{mg}}
	(k\!+\!1)+b- \sigma^{\text{tg}}(k\!+\!1)\\
	\mu^{\text{tg}}(k+1) = \mu^{\text{tg}}(k) + \beta^{\text{tg}} (2\zeta^{\text{tg}}(k+1)-\zeta^{\text{tg}}(k))\\[.5em]
\text{For all busses } y \in \mc B \text{ (in parallel):}\\
\left\lfloor
\begin{array}{l}
\zeta^{\text{pb}}_y(k+1)   = \sum_{i \in \mc P_y^{\mathrm{b}}} p_i^{\mathrm{d}} + \sum_{i \in \mc N_y^{\mathrm{b}}} \eta_i(k\!+\!1)\\
\hspace*{8em} - p^{\text{tg}}_y(k\!+\!1) - \sum_{z \in \mc B_y} p^\ell_{(y,z)} (k\!+\!1)\\[.5em]
\mu_y^{\text{pb}}(k+1) = \mu_y^{\text{pb}}(k) + \beta^{\text{pb}}_y (2 \zeta^{\text{pb}}_y(k\!+\!1)- \zeta^{\text{pb}}_y(k))
\end{array}
\right.
\vspace*{.2em}
\end{array}
\right.
$

\medskip
\hrule
\end{minipage}
\end{figure}

\begin{remark} The main properties of the proposed market-clearing mechanism (Algorithms 1-3) are listed below:
\begin{enumerate}[(i)]

\item The step sizes in the local and central updates (i.e., Algorithms 2 and 3) are fully-uncoordinated, i.e., they can differ across prosumers and DNO, and can be chosen independently based on local information only;

\item The primal update of each prosumer (Algorithm~2~(ii)) involves the solution of a quadratic program\footnote{Up to a fairly-standard reformulation of the absolute value term in \eqref{eq:f_t}.}, for which very efficient solvers are available, e.g. \cite{osqp}. {In there, if $   J_{i} \big( \xi, u_{-i}(k) \big)$ is replaced by its approximate version $ \tilde J_{i} \big( \xi, \sigma^{\text{mg}}(k) \big)$, obtained by neglecting prosumer $i$'s contribution $p_i^{\text{tr}}$ to the aggregative active load $\sigma^{\text{mg}}$, Algorithm 1 will converge to a variational Wardrop equilibrium \cite[\S~II.B]{belgioioso2020semi}, which is a good approximation of v-GNEs for networks with a large number of prosumers.}

\item  The primal update of the DNO (Algorithm~3~(i)) requires projecting onto $\mc U_{N+1}$, which is a convex but nonlinear set. This operation is computationally expensive if naively solved. However, more efficient ad-hoc algorithms to calculate $\proj_{\mc U_{N+1}}$ can be designed using \textit{best approximation methods} \cite[\S~30]{bauschke2011convex}, e.g., see Appendix \ref{APA}.
\item {Algorithm 1 can be recast as a proximal-point method opportunely preconditioned to distribute the computation among the prosumers \cite[\S~IV]{belgioioso2020semi}.
Such operator-theoretic interpretation can be used to design provably-correct acceleration schemes \cite{belgioioso2020semi} as well as to provide robustness guarantees to asynchronous implementations \cite{cenedese2021asynchronous}.
}
\end{enumerate}
\end{remark}
\begin{remark}  {
Our proposed approach differs from community-based local markets \cite[Section 3.2]{sousa2019}, which also requires a coordinator that manages the trading activities. In our setup, each prosumer knows its trading partners and, thus, negotiates directly with them while the coordinator handles the physical constraints and aggregated power bought from the main grid.}
\end{remark}
\section{Numerical Studies}
\label{sub:case_std}
We perform an extensive numerical study on the IEEE 37-bus distribution network to validate the proposed game-theoretic market design and market-clearing algorithm. Specifically: (a) we evaluate the importance of having physical constraints in the model; (b) we evaluate the economical benefits of trading;  (c) we show how storage units owned by prosumers might affect power consumptions; and (d) we test the scalability of the proposed algorithm. All the simulations are carried out in \textsc{Matlab} and use the \texttt{OSQP} solver \cite{osqp} for solving the quadratic programming problems. 

In all simulations\footnote{The codes and data sets used for all simulations are available at \texttt{https://github.com/ananduta/P2Penergy/simulations}.}, we consider heterogeneous networks, where the power demand profile of a prosumer or passive user is either that of single household, multiple households, restaurant, office, hospital, or school. Moreover, some prosumers may have solar-based power generation. The demand and solar-based generation profiles are based on \cite{jasm}.
We also arbitrarily select a set of prosumers to own dispatchable generation units with different sizes and to own homogeneous storage units. We randomly generate the trading networks and place each prosumer and passive user in one of the busses of the IEEE 37-bus network. 

Some of the default cost parameters are set as in \cite{atzeni2013}, i.e., $Q_i^{\mathrm{di}}=0$, $c_i^{\mathrm{di}}=0.045\,  \text{\euro/kW}$, for all $i \in \mc N^{\mathrm{di}}$, $Q_i^{\mathrm{st}}=0$, $c_i^{\mathrm{st}}=0$, for all $i \in \mc N^{\mathrm{st}}$, and $d_h^{\mathrm{mg}}=0.1624/b_h \, \text{\euro/kW}$, whereas the trading cost parameters $c_{(i,j)}^{\mathrm{tr}}=0.08 \,  \text{\euro/kW}$, for all $(i,j) \in \mc E$, and $c^{\mathrm{ta}}=0.01 \,  \text{\euro/kW}$. The parameter $c_{(i,j)}^{\mathrm{tr}}$ is set larger than $c_i^{\mathrm{di}}$ to encourage trading between prosumers with and without dispatchable units, but is smaller than the average unit-price of importing power from the main grid. Note that, in some simulations, we vary these cost parameters. 
\subsection{Achieving operationally-safe solutions}
In the first simulation study, we compare the solutions obtained from solving a P2P market model with and without capacity constraints \eqref{eq:line_cap}. We specifically create an extreme case with $25$ prosumers, where the load of  prosumer $10$ (see Figure \ref{fig:simD_safe}) is very high. We solve both market designs using Algorithm~1. Figure \ref{fig:simD_safe} shows the resulting power-line saturations between busses for both designs. 
Some equilibrium solutions of the P2P market cause overcapacity in some lines when capacity constraints \eqref{eq:line_cap} are not taken into account in the model, as illustrated in Figure \ref{fig:simD_safe}~(b).
\begin{figure}[t]
	\centering
	\begin{subfigure}[t]{0.23\textwidth}
		\centering
		\includegraphics[scale=0.92]{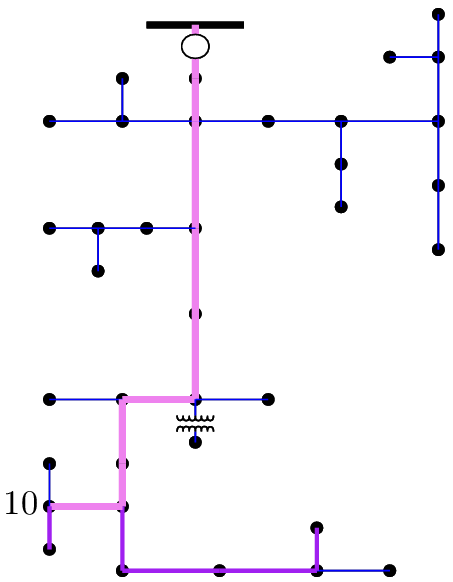}
		\caption{With physical constraints}
	\end{subfigure}%
	~ 
	\begin{subfigure}[t]{0.27\textwidth}
		\centering
		\includegraphics[scale=0.92]{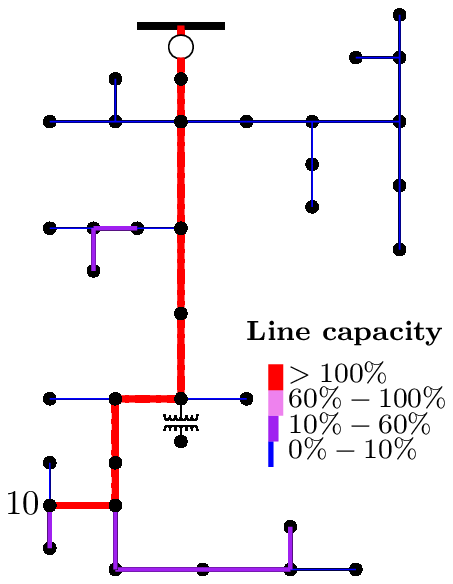}
		\caption{Without physical constraints}
	\end{subfigure}
	\caption{Power line capacities of the physical network. The solutions of the P2P market might cause overcapacity in some lines of the physical network when capacity constraints \eqref{eq:line_cap} are not taken into account.}
	\label{fig:simD_safe}
\end{figure}

\subsection{The impact of P2P trading}
In this section, we evaluate whether energy trading is economically beneficial for the prosumers. To this end, 
we generate a network of 50 prosumers and consider two scenarios: (a) where trading is not allowed, i.e., $\overline{p}_{(i,j)}^{\mathrm{tr}}=0$ in \eqref{eq:p_t_bound}; (b) where trading is allowed with $\overline{p}_{(i,j)}^{\mathrm{tr}}=30 \ \mathrm{kW}$, and the default cost parameters are homogeneous. The other parameters of the network are kept constant in both scenarios. Figure~\ref{fig:simC_cost} shows the individual costs difference between the equilibrium configurations of the market designs with (a) and without P2P tradings (b). In particular, all prosumers gain economical benefits when they can trade.

\begin{figure}[htbp]
	\centering
	\includegraphics[width=.92\linewidth]{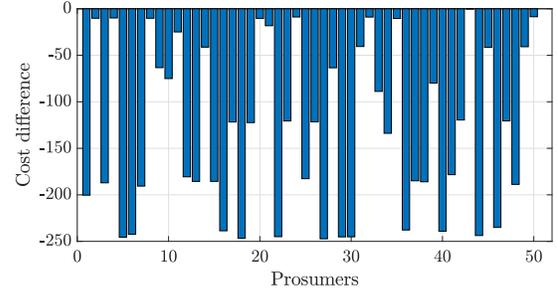}
	\caption{Total cost improvement (\euro) of each prosumer by trading ($c_{(i,j)}^{\mathrm{tr}}=0.08\text{\euro/kW}$).
	}
	\label{fig:simC_cost}
\end{figure}
\begin{figure}[htbp]
	\centering
	\includegraphics[width=.92\linewidth]{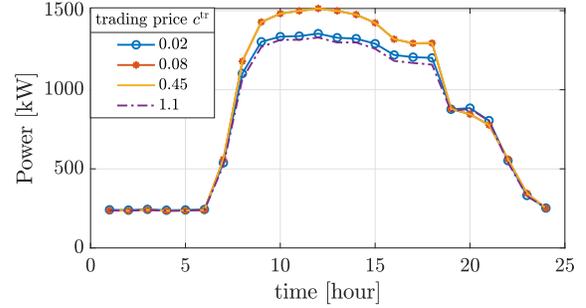}
	\caption{Aggregated P2P trading for different cost coefficients ($c_{(i,j)}^{\mathrm{tr}}$ in \euro/kW).
	}
	\label{fig:simC_ctr}
\end{figure}
\begin{figure}[t!]
	\centering
	\includegraphics[width=.92\linewidth]{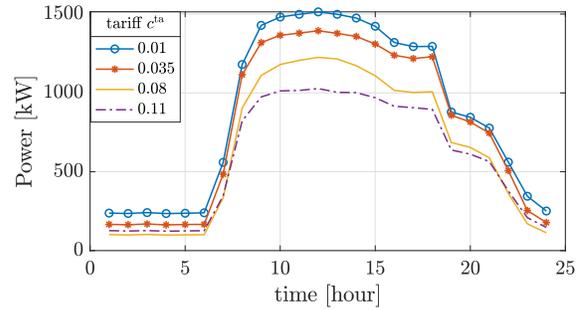}
	\caption{Aggregated P2P trading for different penalty coefficients ($c^{\mathrm{ta}}$ in \euro/kW).
	}
	\label{fig:simC_pen}
\end{figure}
Then, we evaluate the sensitivity of the total traded power with respect to the trading cost parameter $c_{(i,j)}^{\mathrm{tr}}$ and the trading tariff, $c^{\mathrm{ta}}$. Figure \ref{fig:simC_ctr} shows that   $c_{(i,j)}^{\mathrm{tr}}$ must be set appropriately to maximize trading among prosumers. In other words, when $c_{(i,j)}^{\mathrm{tr}}$ is either too high or low, trading is less attractive. 
On the other hand, the higher the tariff is, the less power is traded, as shown in Figure \ref{fig:simC_pen}. Therefore, the DNO may adjust this tariff to encourage or discourage trading in the network. Discouraging trading might be needed when the capacity of the network is close to its limit.

\subsection{The impact of storage units}
In this set of simulations, we investigate the advantages of distributed storage in the network. We generate a test case of 50-prosumer network and consider two extreme scenarios: (a) no prosumers own storage units and  (b) all prosumers own storage units. Furthermore, we also allow some of the prosumers to own distributed generation units, {whose cost functions are strongly convex quadratic, i.e., $Q_i^{\mathrm{di}}>0$, for all $i \in \mc N^{\mathrm{di}}$, which vary from one unit to another}. Figures \ref{fig:sim_B1}-\ref{fig:sim_B3} summarize the simulation results. From Figure \ref{fig:sim_B1}, we can see how the storage units help in shaving the peak of total power imported from the main grid and locally generated by distributed generators. Interestingly, the trading between prosumers is also affected, as shown by Figure \ref{fig:sim_B2}. From this plot, we observe that the existence of storage units reduce the total power traded during the peak hours as the prosumers have reserved energy in their storages. 
Note that the prosumers charge their storage units during the first off-peak hours by buying energy from the main grid and/or from other prosumers that own dispatchable generation units (see the first six hours of the bottom plot of Figure \ref{fig:sim_B1} and those of Figure \ref{fig:sim_B2}).  {Finally, Figure \ref{fig:sim_B3} compares the price of electricity from the main grid and the average price of bilateral trading (including the average of the shadow prices).  Most of the times, the trading prices are lower than the grid prices (in both scenarios), explaining the high amount of power traded.}

\begin{figure}[t]
	\centering
	\includegraphics[width=0.8\linewidth]{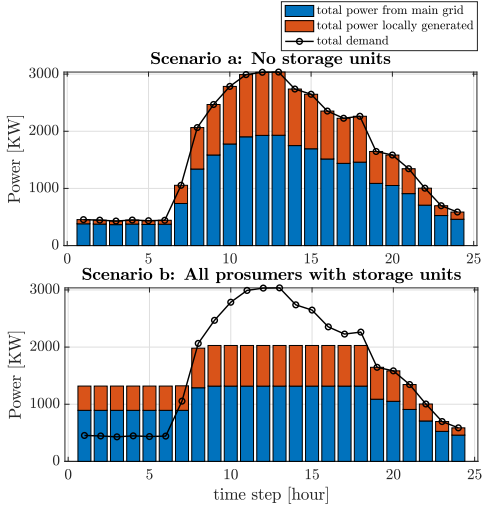}
	\caption{	Incorporating storage units causes a peak-shaving effect on the sum of the total power imported from the main grid and the power locally generated.}
	\label{fig:sim_B1}
\end{figure}
\begin{figure}[t]
	\centering 
	\includegraphics[width=0.92\linewidth]{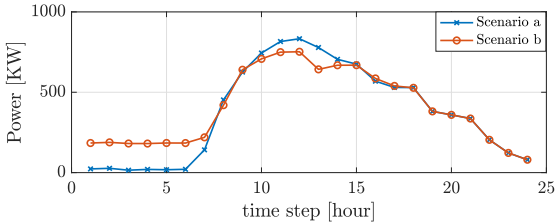}
	\caption{Aggregated P2P trading in scenarios (a) and (b).
	}
	\label{fig:sim_B2}
\end{figure}  
\begin{figure}[t]
	\centering 
	\includegraphics[width=0.92\linewidth]{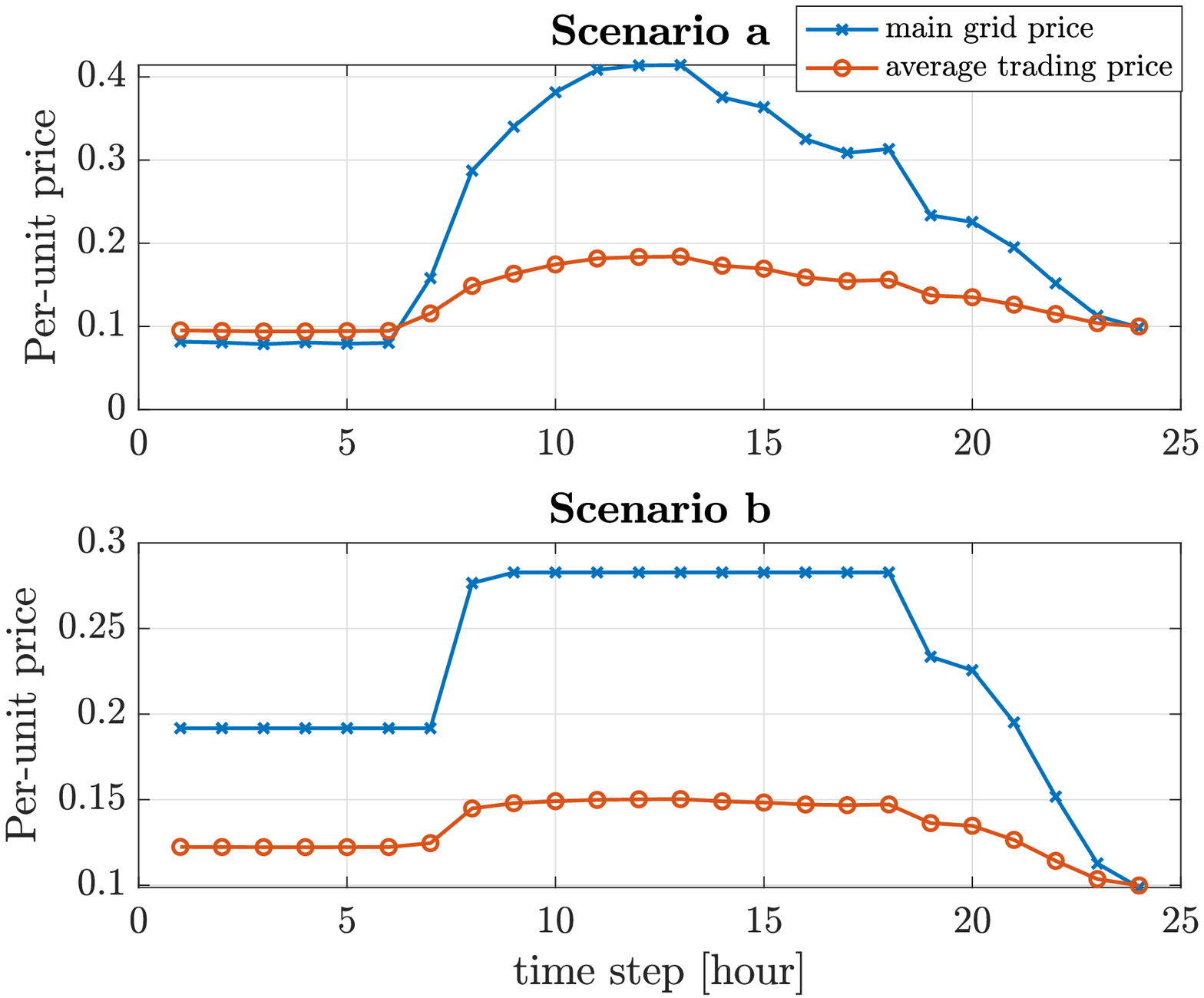}
	\caption{Comparison of the average electricity trading price ($c^{\mathrm{tr}} + c^{\mathrm{ta}} + \tfrac{1}{|\mc E|}\sum_{(i,j)\in \mc E}\mu_{(i,j)}^{\mathrm{tr}}$) with the electricity grid prices.
	}
	\label{fig:sim_B3}
\end{figure}

\subsection{Scalability of the market-clearing mechanism}
Finally, we perform a scalability test for the proposed algorithm. Specifically, we evaluate the convergence speed, in terms of the total number of iterations required to meet a predetermined stopping criterion, when the size of the population of prosumers $N$ and the connectivity of the trading network (the number of trading links) grow. We carry out two sets of simulations.
For the former, we consider five different values of $N$ and a fixed connectivity level of $0.6$ and we run ten Monte Carlo simulations for each $N$, whereas in the latter, the connectivity of the trading network of $50$ prosumers varies in the range $[0.1,1]$, where connectivity $1$ means that the trading network is a complete graph. Similarly, we also run ten Monte Carlo simulations for each connectivity value. We can see from Figure \ref{fig:simE} that Algorithm~1 suitably scales with respect to both the number of prosumers and the connectivity level of the trading network. These results highlight that our algorithm is suitable to be applied to  large-scale systems. 

\begin{figure}[t]
	\centering
	\includegraphics[width=0.9\linewidth]{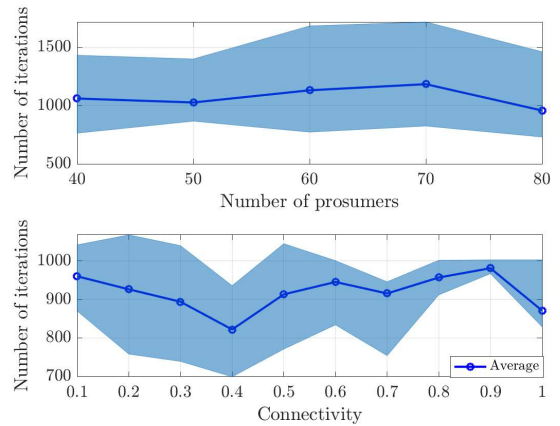}
	\caption{Total number of iterations for convergence of Alg. 1 vs number of prosumers (top) and the connectivity level (the number of trading links) (bottom). {The average computation times of the inner loops, i.e., Algs.~2 and 3, obtained on a computer with Intel Xeon E5-2637 3.5 GHz processors and 128 GB of memory, are 74.5 ms and 1.13 s, respectively.} 
	}
	\label{fig:simE}
\end{figure}

\section{Conclusion}
Energy management and P2P trading in future energy markets of prosumers can be formulated as a generalized game, where the network operator is an extra player in charge of handling the network operational constraints. A provably-convergent operationally-safe market-clearing mechanism is obtained by solving the game with a semi-decentralized Nash equilibrium seeking algorithm based on the proximal-point method.    
Numerical studies show that the computational complexity of the proposed mechanism is independent of the prosumer population size, and suggest that active participation in the market is economically advantageous both for prosumers and network operators.
Future research directions include: efficiently incorporating non-linear convex approximation of power flow in the algorithm; handling the physical constraints in a fully-distributed manner, i.e., without the action of a network operator; and dealing with uncertainties in the model, e.g., renewable energy production, as well as those from information exchange processes required by our algorithm.

\appendix
\subsection{Algorithm 1: Derivation and Convergence Analysis}
\label{app:CA}
The derivation and convergence analysis of Algorithm 1 relies (for the most part) on the \textit{customized preconditioned proximal-point (cPPP)} algorithm for generalized aggregative games proposed in \cite[Algorithm~6]{belgioioso2020semi}. The objective of this appendix is to show that the proposed market-clearing game \eqref{eq:locOpt}, with cost functions and constraints sets defined in \eqref{eq:cost_def}-\eqref{eq:Cconstr_def}, satisfies all the technical conditions in \cite[Theorem~2]{belgioioso2020semi}, among which is the existence of a variational GNE, i.e., item (i) of Proposition \ref{prp:conv}. Therefore, we invoke \cite[Theorem~2]{belgioioso2020semi} to prove convergence of Algorithm 1, i.e., item (ii) of Proposition \ref{prp:conv}. For a complete convergence analysis of the cPPP algorithm for aggregative games we refer to \cite[Appendix~C]{belgioioso2020semi}.

\smallskip
\subsubsection*{Aggregative cost functions} First, we show that the cost functions \eqref{eq:cost_def} can be cast as in\cite[Eqn.~(30)]{belgioioso2020semi}, i.e.
\begin{equation}
	\label{eq:cFcPPP}
	J_i(u_i, u_{-i}) = g_i(u_i) + {(C \avg( u))}^\top u_i,
\end{equation}
where $\avg( u) := \frac{1}{N}\sum_{i \in \mc N} u_i$ denotes the average strategy.
Let $\mc N_i = \mc N$, for all $i \in \mc N$, without loss of generality\footnote{For example, by defining, for all $i \in \mc N$, the ``dummy variables" $\{p^{\text{tr}}_{(i,j)}\}_{j \in \mc N \setminus {\mc N_i}}$ for all the prosumers that do not trade with $i$.}. In this case, $u_i \in \R^{(3+N)H}$, for all $i \in \mc N$. Moreover, let $\Xi^{\text{mg}} \in \R^{H \times (3+N)H}$ denote the matrix that selects the $p_{i}^{\text{mg}}$-component from the decision vectors $u_i$'s, and define the matrix $D:= N \diag(d^{\text{mg}}_1, \ldots, d^{\text{mg}}_H)$, where $d^{\text{mg}}_h$ is the price coefficient for the main grid power. Then, the cost functions in \eqref{eq:cost_def} can be recast as \cite[Eqn.~(30)]{belgioioso2020semi}, or \eqref{eq:cFcPPP}, with
\begin{subequations}
	\label{eq:pieces_cPPP}
	\begin{align}
		g_i(u_i) &= \nonumber
		f_{i}^{\mathrm{di}}(p_{i}^{\mathrm{di}}) + f_{i}^{\mathrm{st}}(p_{i}^{\mathrm{st}})+ f_{i}^{\mathrm{tr}} \left( \{ p_{(i,j)}^{\mathrm{tr}} \}_{j \in \mc N_i} \right),\\[-.2em]
		\label{eq:gi_cPPP}
		&
		\textstyle
		\quad +\frac{1}{N}(Db)^\top p_i^{\text{mg}},\\[.2em]
		C &= (\Xi^{\text{mg}})^\top D \, \Xi^{\text{mg}}.
		\label{eq:C_cPPP}
	\end{align}
\end{subequations}

\subsubsection*{Technical assumptions} Next, we show that all the assumptions in \cite[Theorem~2]{belgioioso2020semi} are satisfied.
\begin{enumerate}[(i)]
	\item For all $i \in \mc N^+$, the cost function $J_i(u_i, u_{-i})$ in \eqref{eq:gi_cPPP} is convex in $u_i$, since all the components of $g_i$ are convex. Hence, \cite[Assumption~1]{belgioioso2020semi} holds.
	\item For all $i \in \mc N^+$, the local set $\mc U_i$ in \eqref{eq:constr_def} is nonempty, closed and convex. Moreover, Slater's constraint qualification on the global feasible set $\left( \prod_{i \in \mc N^+} \mc U_i \right) \cap \mc C$ holds under an appropriate choice of the parameters. Therefore, \cite[Assumption~2]{belgioioso2020semi} is satisfied.
	\item The \textit{pseudo-subdifferential} mapping of the game \eqref{eq:locOpt} reads as $F: u \mapsto \prod_{i \in \mc N}\left(\partial_{u_i} J_i(u_i, u_{-i})\right)\times \0$, since $J_{N+1}=0$. It follows by \cite[Corollary~1]{belgioioso2017convexity}, that the first term of $F$, i.e., $u \mapsto \prod_{i \in \mc N}\partial_{u_i} J_i(u_i, u_{-i})$, is maximally monotone \cite[Definition~20.20]{bauschke2011convex}, since  $C$ in \eqref{eq:C_cPPP} is positive semidefinite, i.e., $C=(\Xi^{\text{mg}})^\top D \, \Xi^{\text{mg}} \succeq 0$. Moreover, also the second term of $F$, i.e., the zero mapping $\0$, is maximally monotone. Therefore, it follows by \cite[Proposition~20.23]{bauschke2011convex} that their cartesian product $\prod_{i \in \mc N}\left(\partial_{u_i} J_i(u_i, u_{-i})\right)\times \0 = F$ is maximally monotone. Hence, \cite[Assumption~6]{belgioioso2020semi} holds.
	\item By \cite[Lemma~1~(i)]{belgioioso2020semi}, there exists a variational GNE of the game in \eqref{eq:locOpt}, since the constraint sets $\mc U_i$ in \eqref{eq:constr_def} are bounded, and the pseudo-subdifferential mapping $F$ is monotone. Hence, \cite[Assumption~4]{belgioioso2020semi} is satisfied.
\end{enumerate} 
{\hfill $\blacksquare$}

\subsection{Alternating Projection for Operational Feasibility}
\label{APA}
In this appendix, we propose an efficient algorithm to compute the projection onto the set $\mathcal{U}_{N+1}$  (Algorithm~3~(i)). First, let us recall the structure of $u_{N+1}$, i.e.,
$$u_{N+1}=\col\left(\{\theta_y, v_y, p_y^{\mathrm{tg}},\{p_{(y,z)}^{\ell},q_{(y,z)}^{\ell}\}_{z \in \mc B_y} \}_{y\in\mc B} \right),$$ and let us define the sets
\begin{align}
	\mc S_1 &:= \{u_{N+1} \mid \text{\eqref{eq:line} and \eqref{eq:grid_ex} hold}\},\\
	\mc S_2 &:= \{u_{N+1} \mid \text{\eqref{eq:pf_p} and \eqref{eq:pf_q} hold}\},
\end{align}
such that $\mathcal{U}_{N+1} = \mc S_1 \cap \mc S_2$. 
The proposed method, summarized in Algorithm~\ref{alg:alg2}, is essentially a \textit{Douglas--Rachford splitting} (DRS) \cite[\S~26.3]{bauschke2011convex} applied to the \textit{best approximation problem} $ {\argmin}_{\xi \in \mc S_1 \cap \mc S_2}  \|\xi-u_{N+1}\| =\proj_{\mc S_1 \cap \mc S_2}(u_{N+1}) $, see e.g. \cite[\S~4.3]{bauschke2015projection} for a formal derivation of the algorithm.

Unlike $\mc U_{N+1}$, the projections onto $\mc S_1$ and $\mc S_2$ have closed-form expressions, hence Algorithm \ref{alg:alg2} only involves elementary operations. Specifically,
$	\proj_{\mc S_1}(u_{N+1})=
u_{N+1}^+$, where
\begin{align*}
	\theta_y^+ &=
	\begin{cases}
		\underline{\theta}_y, & \text{if } \theta_y < \underline{\theta}_y\\
		\overline{\theta}_y, & \text{if } \theta_y > \overline{\theta}_y\\
		\theta_y, & \text{otherwise } \
	\end{cases}, \quad \
	v_y^+ = 
	\begin{cases}
		\underline{v}_y, & \text{if } v_y < \underline{v}_y\\
		\overline{v}_y, & \text{if } v_y > \overline{v}_y\\
		v_y, & \text{otherwise } \
	\end{cases},
	\\
	{p_y^{\text{tg}}}^+ &= 
	\begin{cases}
		p_y^{\text{tg}}, & \text{if } y \in \mc B^{\text{mg}}\\
		0, & \text{otherwise} 
	\end{cases},
\end{align*}
and for all $y \in \mc B$, $z \in \mc B_z$, and  $h \in \mc H$
\begin{align*}
	\begin{array}{l}
		L_{(y,z),h} = \max \left\{ \|\col(p^\ell_{(y,z),h} , q^\ell_{(y,z),h} \|, \, \overline{s}_{(y,z)} \right\},\\[.2em]
		{(p^\ell_{(y,z),h})}^+  = \tfrac{\overline{s}_{(y,z)}}{
			L_{(y,z),h}
		}\,  p^\ell_{(y,z),h}, \\[.2em]
		{(q^\ell_{(y,z),h})}^+  = \tfrac{\overline{s}_{(y,z)}}{
			L_{(y,z),h}
		}\, q^\ell_{(y,z),h}.
	\end{array}
\end{align*}
Whereas, since $\mc S_2$ is an affine set, a closed-form expression for $\proj_{\mc S_2}$ is given in \cite[Example 29.17(ii)]{bauschke2011convex}.

\algblockdefx[init]{init}{Endinit}
[1][<default value>]{\textbf{Initialization}}
[2][<default value>]{\textbf{end initialization}}

\algblockdefx[DNO]{DNO}{EndDNO}
[1][<default value>]{\textbf{DNO routine}}
[2][<default value>]{\textbf{end DNO routine}}

\algblockdefx[PRO]{PRO}{EndPRO}
[1][<default value>]{\textbf{Prosumer $i$ routine}}
[2][<default value>]{\textbf{end prosumer $i$ routine}}

\algblockdefx[Primal]{Primal}{EndPrimal}
[1][<default value>]{\textbf{primal update}}
[2][<default value>]{\textbf{end}}

\algblockdefx[Dual]{Dual}{EndDual}
[1][<default value>]{\textbf{dual update}}
[2][<default value>]{\textbf{end}}

\algblockdefx[Aux]{Aux}{EndAux}
[1][<default value>]{\textbf{auxiliary update}}
[2][<default value>]{\textbf{end}}

\algblockdefx[Comm]{Comm}{EndComm}
[1][<default value>]{\textbf{communication}}
[2][<default value>]{\textbf{end communication}}

\algblockdefx[Agg]{Agg}{EndAgg}
[1][<default value>]{\textbf{aggregation update}}
[2][<default value>]{\textbf{end}}

\algblockdefx[IUC]{IUC}{EndIUC}
[1][<default value>]{\textbf{While convergence is not achieved do:}}
[2][<default value>]{\textbf{end while}}
\medskip 
\begin{minipage}{\columnwidth}
\hrule
\smallskip
\textsc{Algorithm 4}. {DRS for computing $\proj_{\mc U_{N+1}}(u_{N+1})$}
\smallskip
\hrule 
\smallskip
	\begin{algorithmic}[1]
		
		\smallskip
		\State Initialize $\xi(0) \in \bb R^{n_{N+1}}$, and set $\eta \in (0,2)$
		\IUC{ }
		
		\smallskip
		\State
		$z(k) = \proj_{\mc S_1}(\frac{1}{2} \xi(k) + \frac{1}{2} u_{N+1})$

		\smallskip
		\State
		$\xi(k+1) = \xi(k) + \eta \left( \proj_{\mc S_2}    (2z(k)-\xi(k)) - z(k)
		\right)$
		\EndIUC
		
	\end{algorithmic}
	\medskip
\hrule
\end{minipage}

\smallskip	

\bibliographystyle{ieeetran}
\bibliography{ref}


\end{document}